\documentclass[12pt,draftclsnofoot,onecolumn]{IEEEtran}

\usepackage{amsfonts}
\usepackage{times}
\usepackage{latexsym}
\usepackage{amssymb}
\usepackage{amsmath}
\usepackage{cite}
\usepackage{verbatim}

\newcommand{\bydef}{\triangleq}

\def\bydef{:=}

\def\bb0{{\mathbb{0}}}

\def\bydef{:=}

\def\bb{{\mathbf{b}}}

\def\bh{{\mathbf{h}}}

\def\b0{{\mathbf{0}}}

\def\bbE{{\mathbb{E}}}

\def\bbN{{\mathbb{N}}}

\def\bbR{{\mathbb{R}}}

\def\bydef{:=}

\def\sf0{{\mathsf{0}}}

\usepackage{graphicx}
\usepackage{amsbsy}  

\usepackage{amssymb}
\usepackage{amsfonts}
\usepackage{amsmath}
\usepackage{latexsym}
\usepackage{epstopdf}

\begin{document}

\newtheorem{thm}{Theorem}
\newtheorem{lemma}{Lemma}
\newtheorem{rem}{Remark}
\newtheorem{exm}{Example}
\newtheorem{prop}{Proposition}
\newtheorem{defn}{Definition}
\newtheorem{cor}{Corollary}
\def\proof{\noindent\hspace{0em}{\itshape Proof: }}
\def\endproof{\hspace*{\fill}~\QED\par\endtrivlist\unskip}
\def\bh{{\mathbf{h}}}
\def\SIR{{\mathsf{SIR}}}
\def\SINR{{\mathsf{SINR}}}
\def\TH{{\mathsf{TH}}}

\title{Transmission Capacity of Wireless Ad Hoc Networks with Energy Harvesting Nodes}
\author{
Rahul~Vaze
\thanks{Rahul~Vaze is with the School of Technology and Computer Science, Tata Institute of Fundamental Research, Homi Bhabha Road, Mumbai 400005, vaze@tcs.tifr.res.in. }
}
\maketitle
\begin{abstract}
Transmission capacity of an ad hoc wireless network is analyzed when each node of the network harvests energy from nature, e.g. solar, wind, vibration etc. Transmission capacity
is the maximum allowable density of nodes, satisfying a per transmitter-receiver rate, and an outage probability constraint. 
Energy arrivals at each node are assumed to follow a Bernoulli distribution, and each node stores energy using an energy buffer/battery. 
For ALOHA medium access protocol (MAP), optimal transmission probability that maximizes the transmission capacity is derived as a function of the energy arrival distribution. Game theoretic analysis is also presented for ALOHA MAP, where each transmitter tries to maximize its own throughput, and symmetric Nash equilibrium is derived. For CSMA MAP, back-off probability and outage probability are derived in terms of input energy distribution, thereby characterizing the transmission capacity.
\end{abstract}
\section{Introduction}
Consider an ad hoc wireless network, where multiple source-destination pairs try to communicate without the help of a 
centralized controller. There are many examples of real life ad hoc networks, such as military networks, vehicular networks, sensor networks, etc. Typically, each node of an ad hoc network is powered by a conventional battery, that has limited lifetime and needs to be replenished periodically. Recently, to improve the lifetime of nodes, and to provide a means of {\it green} communication, the concept of equipping nodes with energy harvesting devices that can extract or tap energy from renewable energy sources, such as solar, wind, vibration, etc., has been introduced. Harvesting energy from nature, however, makes the future energy availability random, and the transmission strategies have to adapt dynamically to the energy arrivals.

In this paper, we are interested in characterizing the  transmission capacity of an ad hoc network when each transmitter has energy harvesting capability. The transmission capacity of an ad hoc network 
is the maximum allowable density of nodes, satisfying a per transmitter-receiver rate, and outage probability constraint \cite{Weber2005, Baccelli2006}. Similar to \cite{Weber2005, Baccelli2006}, we assume that the locations of transmitters are distributed as a homogenous 
Poisson point process (PPP) with a fixed density.
We assume that energy arriving to each node through natural sources follows  a Bernoulli distribution, 
where in each discrete time, a unit amount of energy arrives with probability $p$ or no energy arrives with probability $1-p$. 
Energy arrivals are independent and identically distributed across all transmitters in the network. We remark in the sequel that even though we consider Bernoulli energy arrival model, our results apply to general energy arrival distributions with unbounded battery capacity.

In this paper we consider two medium access protocols (MAPs):  ALOHA, and CSMA, to be used by each transmitter. For ALOHA MAP, we derive the optimal transmission probability as a function of energy harvesting rate $p$ that maximizes the transmission capacity of the ad hoc network. We also consider a selfish setting, where each transmitter is interested in maximizing its own throughput, and derive symmetric Nash equilibrium strategies. For the CSMA MAP, we derive the back-off probability and the outage probability that together characterize the transmission capacity. Before describing our contributions, we first review some recent work on the design of wireless networks with energy harvesting.
 \subsection{Prior Work}
 Assuming non-causal knowledge of future energy arrivals, optimal offline power allocation strategies to maximize throughput have been derived for single source-destination pair \cite{UlukusEH2011c}, interference channel \cite{YenerIntChan2011}, and broadcast channel \cite{Uysal2011}. For a single source-destination pair, structure of a causal throughput optimal strategy with energy arrival distribution knowledge has been derived in \cite{ChaporkarEH2011}, while a distribution free  online algorithm has been derived in \cite{VazeEH2011}. For a wireless network with energy harvesting nodes, the probability that a node 
 successfully transmits a data packet to a fusion center is analyzed for time division multiple access and ALOHA MAP in \cite{Simeone2011}. From a queuing theoretic point of view, queue stabilizing policies have been derived for a single source-destination pair \cite{VSharmaEH2010}, and a two-user interference channel \cite{EphremidesEH2011}.
 
The most relevant work to this paper is \cite{HuangEH2011}, where each node of the ad hoc network harvests energy with arbitrary energy arrival distribution with fixed rate of arrival $p$. Each transmitter is scheduled to transmit with power $P$ if it 
has more than $P$ amount of energy irrespective of all other nodes. Optimal value of $P$ is derived in \cite{HuangEH2011} that maximizes the transmission capacity as a function of network parameters.
Compared to \cite{HuangEH2011}, in this paper we take a different viewpoint and couple the energy queue dynamics with the ALOHA/CSMA transmission probability, and try to find the best ALOHA transmission probability that maximizes the transmission capacity. In our model, 
a transmitter sends a packet with probability $q$ with ALOHA (if channel is sensed idle with CSMA) if there is energy available at the transmitter.

For the selfish setting, where each transmitter of an ad hoc network powered with a conventional power source tries to unilaterally maximize its own throughput, optimal transmission probability for ALOHA MAP has been derived  in \cite{Hanawal2012}. 
Not surprisingly an {\it always transmit strategy} is selfishly optimal, and to improve the selfish behavior towards globally optimal solution, a penalty function is introduced in the objective function that linearly increases the penalty with the transmission probability of each node.  
 
 \subsection{Contributions}
\begin{itemize}
\item For ALOHA MAP, we find the optimal transmission probability that each transmitter should use to maximize the  transmission capacity of the ad hoc network for both finite and infinite battery capacity. For the infinite battery capacity case, the optimal transmission probability takes two values depending on a threshold that is a function of the network parameters. For finite battery capacity case, the optimal transmission probability is shown to satisfy the probability of having non-zero energy in the energy queue to be equal to a constant, with explicit solution found for the unit battery capacity case.
\item For ALOHA MAP, we find the optimal transmission probability that each transmitter should use that maximizes its own throughput (selfish setting), and characterize symmetric Nash equilibrium (SNE). At SNE, each transmitter uses 
transmission probability $q$ such that $p\le q\le 1$ with infinite battery capacity, and $q=1$ with finite battery capacity. 
It is shown that the price of anarchy, that is ratio of globally optimal transmission capacity to transmission capacity obtained at the worst SNE, is quite low and is actually equal to {\it one} for some cases. Thus, in contrast to \cite{Hanawal2012}, the selfishly optimal strategy in the energy harvesting setting is not much different from the globally optimal strategy,  since with energy harvesting nodes, individual objective functions are 
inherently energy aware.  
\item For CSMA MAP, we derive the back-off probability and the outage probability for each transmitter when each transmitter harvests energy from renewable sources, consequently the transmission capacity.
\end{itemize}

We remark that using a simple Bernoulli energy arrival model and energy queue based transmission probability framework, this work
takes few initial steps towards understanding the fundamental performance of energy harvesting ad hoc networks, and a lot more work is required for complete characterization.

\section{System Model for ALOHA MAP}\label{sec:sys}
Consider a wireless ad hoc network where multiple source destinations pairs want to communicate with each other without any centralized control. Following \cite{Weber2005}, the location of transmitter nodes ${\cal T}_m, \ m\in \bbN$ is assumed to be distributed  as a homogenous Poisson point process (PPP) $\Phi = \{{\cal T}_m\}$ on a two-dimensional plane with density $\lambda$ \cite{Stoyan1995}. The receiver  $R_m$ associated with transmitter $T_m$ is assumed to be at a fixed distance of $d$ from $T_m$ with an arbitrary orientation.

We assume that each transmitter harvests energy from nature, e.g. through solar, wind, peizo-eletric sources etc. Each transmitter is assumed to have a battery of capacity $B$ using which it can store the harvested energy. The energy arrival process is assumed to be i.i.d. Bernoulli with rate $p$ across different transmitters, i.e. at each time $t$, either a unit amount of energy arrives at any transmitter with probability $p$, or no energy arrives with probability $1-p$. The results of this paper are applicable to general energy arrivals as pointed out in Remark \ref{rem:genenergy}. The Bernoulli energy arrival assumption is made for simplicity of exposition and for deriving critical insights into the problem.

If $x_m$ is the transmitted signal from $T_m$ at time $t$, then the received signal $y_m$ at $R_m$ is given by 
\begin{eqnarray}\label{eq:rxsig}
y^t_m = \sqrt{P} d^{-\alpha/2}h_{mm} x_m + \sum_{{\cal T}_s \in \Phi  \backslash \{ {\cal T}_0\} }\sqrt{P}{\mathbf 1}_{T_s}(t)d_{ms}^{-\alpha/2}h_{ms} x_s + z_m,
\end{eqnarray}
where $P$ is the power transmitted by each transmitter, $\alpha >2$ is the path-loss exponent, ${\mathbf 1}_{T_s}(t)$ is the indicator function that is $1$ if $T_s$ transmits at time $t$ and is zero otherwise, $h_{k\ell}$ is the fading channel coefficient between transmitter $T_{\ell}$ and receiver $R_k$ that is assumed to be Rayleigh distributed, and $z_m$ is the zero mean unit variance additive white Gaussian noise.

We consider the interference limited regime, i.e. noise power is negligible compared to the interference power, and
henceforth drop the noise contribution \cite{Weber2005}.\footnote{This assumption is made for simplicity of exposition, and all results of this paper can be easily extended to the case of  additive noise as well.} With the interference limited regime, we assume unit power transmission, $P=1$, since the signal to interference ratio (SIR) does not depend on $P$. Let $\SIR_m(t)$ denote the SIR between $T_m$ and $R_m$ at time $t$, then using (\ref{eq:rxsig}) 
$$\SIR_m(t) \bydef \frac{ 
d^{-\alpha} |h_{mm}|^2
}
{
\sum_{
{\cal T}_s \in \Phi \backslash \{ {\cal T}_m\}
}
{\mathbf 1}_{T_s}(t) 
d_{ms}^{-\alpha}|h_{ms}|^2}.$$

We consider a slotted ALOHA like random access MAP, where each transmitter attempts to transmit its packet with an access probability $q$, independently of all other transmitters if it has energy to transmit. Thus, with transmit power $P=1$, 
any transmitter attempts to transmit if it has non-zero amount of energy. Let $E_m(t)$ be the amount of energy available with $T_m$ at time $t$. Then $E_m(t)$ is an i.i.d. birth-death Markov process with the following transition probabilities, 
$P(E_m(t+1) = E_m(t) +1) = p(1-q)$,  $P(E_m(t+1) = E_m(t) -1) = q(1-p)$, $P(E_m(t+1) = 1 | E_m(t) = 0) = p$ and $P(E_m(t+1) = 0 | E_m(t) = 0) = 1-p$ as shown in Fig. \ref{fig:BD}. 
Let $r \bydef P (E_m(t)\ge1)$ be the probability of $E_m(t)\ge 1$. Let $\chi_m(t)=1$ if  $T_m$ transmits at time $t$ given that $E_m(t)\ge1$, and $\chi_m(t)=0$ otherwise. Thus, $P( \chi_m(t)=1 | E_m(t)\ge1 ) =q$.
By definition, any transmitter transmits $({\mathbf 1}_{T_s}(t)=1)$ with probability $P(E_m(t)\ge1, \chi_m=1) = P(E_m(t)\ge1) P( \chi_m(t)=1 | E_m(t)\ge1 ) = P(E_m(t)\ge1) q = rq$, independently of all other nodes. Consequently, the active transmitter process is a homogenous PPP on a two-dimensional plane with density $\lambda_a \bydef q r \lambda$.

We define that the transmission from $T_m$ to $R_m$ at time $t$ is successful if the SIR between $T_m$ and $R_m$ is greater than a threshold $\theta$, i.e. $\SIR_m(t) > \theta$.
Thus, the probability of success at time $t$ is defined to be 
\begin{equation}\label{eq:sucevent}
P_{suc}(t) = P ( \SIR_m(t) > \theta),
\end{equation}
We assume that the rate of transmission corresponding to threshold $\theta$ is $R= \log(1+ \theta)$ bits/sec/Hz. Then the transmission capacity \cite{Weber2005} is defined to be $ C = \lambda_a P_{suc}(t) R \ \text{bits/sec/Hz/m}^2$.

Our goal is find the optimal $q$ that maximizes the transmission capacity, $q^{\star} = \arg \max_q C$. 
For the purpose of analyzing the success probability $P_{suc}$ and $C$, we consider a typical transmitter receiver pair $T_m, R_m$. It has been shown in \cite{Weber2005} that for the PPP distributed transmitter locations, the performance of the typical source destination pair is identical to the network wide performance using Slivnyak's Theorem \cite{Stoyan1995}. Next, we first analyze the case when each transmitter has an unbounded battery capacity $B=\infty$, and later extend it to the finite battery capacity case.

\subsection{$B=\infty$ (Infinite Battery Capacity)}
From Fig. \ref{fig:BD}, using well-known analysis of the infinite state birth-death Markov process we have that $r = P(E_m(t) \ge 1) = \min\left\{\frac{p}{q}, 1\right\}$ for $p, q >0, \ \forall \ m$. Throughout this paper, the only quantity we will be interested in from the energy queue point of view is the probability that the energy queue is not in state $0$, $P(E_m(t) \ge 1)$. Next, we remark that because of this restricted dependence, the results of this paper generalize to any energy arrival distribution, and are not restricted to just Bernoulli distribution. 
\begin{rem}\label{rem:genenergy} Let the energy arrival process be $A(t)$ with an arbitrary distribution of rate $p$ that is independent and identically distributed for all transmitters. Let each transmitter transmit with probability $q$, if there is more than $P$ amount of energy in the battery, and uses power $P$ for its transmission (similar to the ALOHA setup described before). 
The energy queue dynamics at any transmitter is given by $E(t+1) = E(t) + A(t) - P {\mathbf 1}_{E(t)\ge P} \chi(t)$, where $q= \bbE\{\chi(t)| E(t)\ge P\}$. For this general model with $B=\infty$, it has been shown in \cite{HuangEH2011} that $P(E(t) \ge P) =  1$ if $p \ge Pq$, and  $P(E(t) \ge P)=  \frac{p}{Pq}$ if $p < Pq$. In \cite{HuangEH2011}, the result is obtained for $q=1$, but it can be readily generalized for any $q \in [0,1]$. Specializing, this result for our model with unit power $P=1$, we get that $P(E(t) \ge 1) = \min\left\{\frac{p}{q}, 1\right\}$ for any energy arrival distribution. Thus, we do not lose any generality by restricting ourselves to the Bernoulli energy arrival distribution. 
\end{rem}

Consider the success probability (\ref{eq:sucevent}),
\begin{eqnarray*}
P_{suc} &= &P (\SIR_m(t) > \theta) ,\\
&=& P\left(\frac{ 
d^{-\alpha} |h_{mm}|^2
}
{
\sum_{
{\cal T}_s \in \Phi \backslash \{ {\cal T}_0\}
}
{\mathbf 1}_{T_s}(t) 
d_{ms}^{-\alpha}|h_{ms}|^2} > \theta\right),\\
&=&  \exp\left(-\lambda r q d^2 \theta^{\frac{2}{\alpha}}\kappa(\alpha)\right),  \mbox{ \cite{Baccelli06} \ \text{where} \ $\kappa(\alpha) = \frac{2\pi^2}{\alpha sin(2\pi/\alpha)}$}.
\end{eqnarray*}
Let $\lambda_{max} \bydef \frac{1}{d^2 \theta^{\frac{2}{\alpha}}\kappa(\alpha)}$.
Hence the transmission capacity is $ C = \lambda r q   \exp\left(-\frac{r q \lambda}{\lambda_{max}}\right)R$. 
We derive the optimal transmission probability $q^{\star}$ in Theorem \ref{thm:gloptp}. We need the following definition for the proof of Theorem \ref{thm:gloptp}.
\begin{defn} A function $f: \bbR \rightarrow \bbR$ is called {\it unimodal} if for some value $m$, it is monotonically increasing for 
$x\le m$ and monotonically decreasing for $x > m$.
\end{defn}

\begin{thm}\label{thm:gloptp}
The optimal ALOHA transmission probability with $B =\infty$ that maximizes the transmission capacity 
is $q^{\star} = \frac{\lambda_{max}}{\lambda}$ if $p > \frac{\lambda_{max}}{\lambda}$, and any $q \in [p, 1]$ is optimal otherwise.
\end{thm}
\begin{proof}
With $r = \min\left\{\frac{p}{q}, 1\right\}$,  for $\frac{p}{q} < 1$, the transmission capacity is  
$C = \lambda p  \exp\left(-\frac{p \lambda}{\lambda_{max}}\right)R$, while with $\frac{p}{q} \ge 1$, where $r=1$,
$C = \lambda  q   \exp\left(-\frac{q \lambda}{\lambda_{max}}\right)R$.


Case 1: $\left\{\frac{\lambda_{max}}{\lambda} <1 \ \text{and} \ p \ge \frac{\lambda_{max}}{\lambda}.\right\}$ Note that  $\frac{\lambda_{max}}{\lambda} = \arg \max_{q}\lambda  q   \exp\left(-\frac{q \lambda}{\lambda_{max}}\right)R$ \cite{Baccelli06}. Moreover, since 
$\frac{\lambda_{max}}{\lambda} <1$ and $p \ge \frac{\lambda_{max}}{\lambda}$, considering $0\le q\le p$, where $r=1$,
$\frac{\lambda_{max}}{\lambda} =  \arg \max_{0\le q\le p} C =\arg \max_{0\le q\le p}  q   \exp\left(-\frac{q \lambda}{\lambda_{max}}\right)R $, since $q=\frac{\lambda_{max}}{\lambda}$ lies in the  feasible set $0\le q\le p$. With  
$q^{\star} =  \frac{\lambda_{max}}{\lambda}$, the optimal transmission capacity is $C =\frac{\lambda_{max}}{e}$. 

{Case 2: Let $p < \frac{\lambda_{max}}{\lambda}$.}  Since $C = \lambda q  \exp\left(-\frac{q \lambda}{\lambda_{max}}\right)R$ is a unimodal function of $q$, and achieves its maxima at $\frac{\lambda_{max}}{\lambda}$, it implies that 
$\lambda q  \exp\left(-\frac{q \lambda}{\lambda_{max}}\right)R$ is an increasing function of $q$ for $0\le q \le p$ for $p < \frac{\lambda_{max}}{\lambda}$. Hence $\max_{0\le q \le p } C=\lambda p  \exp\left(-\frac{p \lambda}{\lambda_{max}}\right)R$. Moreover, for any $q>p$, 
$C = \lambda p  \exp\left(-\frac{p \lambda}{\lambda_{max}}\right)R$. Hence, if $p < \frac{\lambda_{max}}{\lambda}$, then any $q\in [p, 1]$ is optimal.
%
\end{proof}

Next, we consider the more realistic case of finite battery capacity $B$ at each node.

\subsection{Finite $B$ (Finite Battery Capacity)}
With finite battery capacity $B$, the energy queue transition probabilities are illustrated in Fig. \ref{fig:maxBD}, and let $r_B \bydef P(E_m(t) \ge 1)$. With finite battery capacity $B$, the energy queue is a  finite state birth-death Markov process for which we have that for $B=1$,
$r_1 = P(E_m(t) \ge 1) = \frac{p}{p+q-pq}$, while for $B>1$, $r_B = \frac{\frac{p}{q}\left(1-\left(\frac{p(1-q)}{q(1-p)}\right)^B\right)}{1-\frac{p}{q}\left(\frac{p(1-q)}{q(1-p)}\right)^B}$ for $p, q >0, p\ne q$ and $r_{B} = \frac{B}{B+1-p}$ for $p=q$. 
Similar to $B =\infty$ case, the transmission capacity expression for finite $B$ is 
$C = \lambda r_B q   \exp\left(-\frac{\lambda r_B q}{ \lambda_{max}}\right),$ and we want to maximize $C$ with respect to $q$. 
We first prove this intermediate result that is important for subsequent analysis.
\begin{lemma}\label{lem:inc} For finite $B$, function $f_{B}(q)\bydef qr_B$ is an increasing function of $q$ for $q \in [0,1]$.
\end{lemma}
\begin{proof}
See Appendix \ref{app:nondec}.\end{proof}

\begin{thm}\label{thm:gloptparbB}
The optimal ALOHA transmission probability with finite battery capacity $B$ is $q^{\star} = \min\left\{{\hat q}, 1\right\}$, where 
${\hat q}$ is the solution to the equation $ f_{B}(x) = \frac{\lambda_{max}}{\lambda}$.
\end{thm}

\begin{proof} From Lemma \ref{lem:inc}, we know that $f_B(q)$ is an increasing function of $q$ for $q \in [0, 1]$. Hence for fixed $\lambda$ and $\lambda_{max}$, 
$C = \lambda f_B(q)\exp\left(- \frac{\lambda} {\lambda_{max}} f_B(q)\right)$ is a unimodal function of $q$. Thus, either  the maxima of $C$ lies in  $[0,1]$, or $C$ is an increasing function for $[0,1]$. To find the maxima of $C$, we equate the first derivative of $C$, 
$C' = \lambda f_B'(q) \exp\left(- \frac{\lambda} {\lambda_{max}} f_B(q)\right) - \lambda f_B(q) \frac{\lambda} {\lambda_{max}} f_B'(q) \exp\left(- \frac{\lambda} {\lambda_{max}} f_B(q)\right)$ to zero, which 
yields $\frac{\lambda f_B(q)}{\lambda_{max}} = 1$ and $f_B(q) = \frac{\lambda_{max}}{\lambda}$. Let ${\hat q}$ solve 
$f_B(q) = \frac{\lambda_{max}}{\lambda}$. If ${\hat q}>1$, then $q^{\star}=1$, otherwise $q^{\star}={\hat q}$.
\end{proof}


\begin{cor}\label{cor:gloptp}
With unit battery capacity $B =1$, $q^{\star} = \min\left\{\frac{ p \lambda_{max}}{\lambda p + \lambda_{max} (1-p)}, 1\right\}$.
\end{cor}


\begin{proof} From Lemma \ref{lem:inc} for $B=1$, $f_1'(q)  > 0$ for $q \in [0,1]$. Hence using  Theorem \ref{thm:gloptparbB}, we know that $q^{\star}$ satisfies $\frac{\lambda f_1(q^{\star})}{\lambda_{max}} = 1$, where $f_1(q) = \frac{pq}{p+q-pq}$. Consequently $q^{\star} = \frac{ p \lambda_{max}}{\lambda p + \lambda_{max}(1- p)}$.
\end{proof}

{\it Discussion:} In this section, we derived the optimal transmission probability for ALOHA MAP when each transmitter is harvesting energy with a rate $p$ Bernoulli distribution. We also remarked that our results apply to any general energy arrival distribution. 
For the infinite battery capacity case, we showed that the optimal transmission probability takes two values depending on a  threshold that is function of the network parameters. 
There is a natural connection between the optimal transmission probability with nodes that are powered with energy harvesting sources and conventional power sources \cite{Baccelli2006}. If rate $p$ is more than the optimal transmission probability with conventional power sources, then we know that the probability that there is  more than unit energy in the queue is $1$. Hence there is always energy to transmit, and the optimal transmission probability with energy harvesting case is equal to the conventional power sources.
If rate $p$ is less than the optimal transmission probability with conventional power sources, 
then we show that transmission capacity is an increasing function of transmission probability from $0$ to $p$ and then becomes a constant for any transmission probability greater or equal to $p$. Thus, any transmission probability greater or equal to $p$ is shown to be optimal.

We also derived results for the realistic case of  finite $B$. We showed that the transmission capacity expression is unimodal, and its maxima as a function of transmission probability 
can be found by solving the probability of non-zero energy in queue to be equal to a constant. 
For the special case of $B=1$, we derived explicit solution for the optimal transmission probability. 

In this section, we considered maximizing system wide transmission capacity that captures the sum throughput of the network. 
In the next section, we consider the case when each transmitter selfishly tries to maximize its own throughput and derive symmetric Nash equilibrium strategies.


\section{Selfishly Optimal Transmission Strategy}
In this section, we consider the case when each transmitter wants to maximize his own throughput, and derive the selfishly optimal transmission probability for each transmitter using ALOHA MAP. 
The energy harvesting and transmission protocol at each transmission is assumed identical to the previous section except for the transmission probability $q$, which is now transmitter dependent, and we denote the transmission probability of transmitter $T_n$ by $q_n$. Therefore the objective function (utility) that each transmitter maximizes is it own throughput $\TH_n \bydef r_n q_n P ( \SIR_n(t) > \theta) R$, where $r_n = P(E_n(t) \ge 1)$. 

In this selfish setting, we will consider symmetric Nash equilibrium (SNE), at which all nodes use the same transmission probability $q^{*}$, since no two transmitters are distinguishable from each other. For more details see \cite{MackenzieGT2003}.
In this setting, let the $n^{th}$ transmitter use $q_n= q$ ($r_n=r$), while the $m^{th}$ transmitter uses $q_m = {\tilde q}, m \ne n$ ($r_m ={\tilde r}$). Then we denote the throughput of the $n^{th}$ transmitter as $\TH_n(q,{\tilde q}) \bydef r q P ( \SIR_n(t) > \theta ) R$. Then $q^{*}$ is a SNE if for each transmitter 
$\TH(q^{*}, q^{*}) = \max_{q \in [0, 1]} \TH(q, q^{*})$. The transmission capacity of the ad hoc network is defined as the sum of 
 throughput of all nodes $C = \lambda  \TH(q^{*}, q^{*})R$.

\subsection{Infinite Battery Capacity $B=\infty$}
\begin{thm}\label{thm:SNE} In the infinite battery capacity case ($B = \infty$), any $q^{*}$ such that $p \le q^{*} \le 1$ is a SNE, where $p$ is the rate of energy arrivals. Moreover, at any SNE, each transmitter gets the same throughput $\TH= p\exp\left(-\frac{p \lambda}{\lambda_{max}}\right)$, and the transmission capacity is $C = \lambda p\exp\left(-\frac{p \lambda}{\lambda_{max}}\right)$.
\end{thm}
\begin{proof} Consider 
\begin{eqnarray*}
\TH_n(q,{\tilde q}) & =& r q P ( \SIR_n(t) > \theta ), \\
&=& r q P\left(\frac{ 
d^{-\alpha} |h_{mm}|^2
}
{
\sum_{
{\cal T}_s \in \Phi \backslash \{ {\cal T}_0\}
}
{\mathbf 1}_{T_s}(t) 
d_{ms}^{-\alpha}|h_{ms}|^2} > \theta\right), \ \ \text{where} \ {\mathbf 1}_{T_s}(t) =1\  \text{ w.p.} \  {\tilde r}{\tilde q},  \\
&=&r q \exp\left(- \frac{{\tilde r}{\tilde q}\lambda}{\lambda_{max}}\right).
\end{eqnarray*}
We note that the throughput of the $n^{th}$ transmitter is monotone increasing in $r q$. Hence we need to find the optimal $q$ that maximizes $r q$. From the previous section, we know that for the infinite battery capacity case ($B = \infty$), 
$r = \min \left\{\frac{p}{q}, 1\right\}$. Thus, if $q < p$, then $r =1$ and $r q = q$,
while if $q \ge p$, $r =\frac{p}{q}$ and $rq =p$.
Thus, any $q$ such that $p \le q \le 1$, maximizes $rq$, and provides the selfishly optimal throughput for the $n^{th}$ transmitter. Therefore, each $q^{*}$ such that $p \le q^{*} \le 1$ is a SNE, and the throughput of each transmitter for $p \le q^{*} \le 1$ is $\TH =p \exp\left(-  \frac{p\lambda}{\lambda_{max}}\right)$.
\end{proof}
Since any $q$ such that $p \le q\le 1$ is a SNE, it might appear that some transmitters can use $q=1$, and 
create more interference for other transmitters. Each transmitter, however, gets to transmit only with probability 
$rq = p$ for $p \le q\le 1$, since if $q$ is large, transmitter uses up more energy and there is higher chance of energy queue to be in state $0$.
Next, we compute the price of anarchy that compares the performance loss incurred due to selfishness by each transmitter. 

\begin{defn} The price of anarchy (PoA) of a game is the ratio of the utility at the globally optimal solution to the utility at the worst equilibrium.
\end{defn}

\begin{lemma}\label{lem:PoA} The PoA of the throughput game is $1$ if $p < \frac{\lambda_{max}}{\lambda}$, and 
$\frac{\lambda_{max}}{e  p \lambda \exp\left(\frac{-p\lambda}{\lambda_{max}}\right)}$, otherwise.
\end{lemma}
\begin{proof} If $p < \frac{\lambda_{max}}{\lambda}$, the globally optimal transmission probability $q^{\star}$ is that such that $p \le q^{\star} \le 1$ (Theorem \ref{thm:gloptp}), which is identical to the selfishly optimal transmission probability $q^*$ achieving the 
SNE (Theorem \ref{thm:SNE}). Hence PoA is $1$ when $p < \frac{\lambda_{max}}{\lambda}$. 
For $p > \frac{\lambda_{max}}{\lambda}$, the globally optimal transmission probability is $q^{\star} = \frac{\lambda_{max}}{\lambda}$, whereas the selfishly optimal policy is  $p \le q^{*} \le 1$. The required expression is obtained by taking the ratio of 
globally optimal transmission capacity expression and transmission capacity expression at any SNE.
\end{proof}

\begin{rem} Recently, selfishly optimal transmission probability for ALOHA MAP has been derived  in \cite{Hanawal2012} for throughput maximization ($\TH(q,{\tilde q})$ defined above) with PPP distributed transmitters, 
where each transmitter is powered by a conventional energy source. As expected, without any energy constraints, $q=1$ (always transmit strategy) is selfishly optimal, however, it provides a very poor PoA performance. 
To improve the selfish behavior towards globally optimal solution, a modified objective function 
is considered $\TH(q,{\tilde q}) - \rho q$ in  \cite{Hanawal2012} that linearly penalizes the increase in transmission probability $q$, where  $\rho$ is a constant. It has been shown that by carefully choosing the scaling parameter $\rho$, the PoA can be significantly improved. 

In comparison to \cite{Hanawal2012}, for the setting in this paper, where each transmitter harvests energy from nature and stores it in a battery, the objective function $\TH(q,{\tilde q})$ is already energy aware, and no extra energy dependent factors need to be introduced for obtaining good PoA performance. 
More importantly, we notice that PoA obtained with the improved energy penalty strategy \cite{Hanawal2012} is worse than the PoA obtained with our model of energy queue dependent transmission probability. Thus, even for conventioanal energy powered sources, using a virtual energy queue based transmission probability can improve the PoA performance.
\end{rem}

\subsection{Finite $B$}
\begin{thm}\label{thm:SNEfinite} With finite battery capacity $B$, $q^{*}=1$ is a SNE. At any SNE, each transmitter gets the same throughput $\TH= p \exp\left(-\frac{p\lambda}{\lambda_{max}}\right)$ and transmission capacity is $C =\lambda p \exp\left(-\frac{p\lambda}{\lambda_{max}}\right)$.
\end{thm}
\begin{proof} Similar to the $B =\infty$ case, we have that  
$\TH_n(q,{\tilde q}) = r q \exp\left(- {\tilde r} {\tilde q} \exp\left(\frac{\lambda}{\lambda_{max}}\right) \right)$,
and  we need to find the optimal $q$ that maximizes $r q$, i.e. find $q$ that maximizes $rq = f_B(q)$. 
The function $f_B(q)$ is an increasing function of $q$ (Lemma \ref{lem:inc}), and achieves its optimal value equal to $p$ at $q=1$. Thus $q^{*} =1$ is a SNE with throughput of each transmitter $\TH =p \exp\left(- \frac{p\lambda}{\lambda_{max}}\right)$.
\end{proof}

Next, we compute the PoA for the finite  battery capacity case. 

\begin{lemma}\label{lem:PoA} The PoA of the throughput game with finite battery capacity  is $1$ if $y^*>1$, where $y^*$ is the solution to 
$f_{B}(y) = \frac{\lambda_{max}}{\lambda}$, and 
$\frac{\lambda q^{\star} \exp\left(- \frac{ q^{\star} \lambda}{\lambda_{max}}\right)}{  p \lambda \exp\left(- \frac{p\lambda}{\lambda_{max}}\right)}$ where $f_{B}(q^{\star}) = \frac{\lambda_{max}}{\lambda}$, otherwise.
\end{lemma}
\begin{proof} If $y^*$ satisfies $f_{B}(y)= \frac{\lambda_{max}}{\lambda}$ and $y>1$, then the globally optimal transmission probability  $q^{\star}=1$ (Theorem \ref{thm:gloptparbB})  is equal to the selfishly optimal transmission probability $q^{*}=1$ (Theorem \ref{thm:SNEfinite}).
Otherwise, the globally optimal transmission capacity is $\lambda q^{\star} \exp\left(- \frac{ q^{\star} \lambda}{\lambda_{max}}\right),$  where $f_{B}(q^{\star}) = \frac{\lambda_{max}}{\lambda}$ (Theorem \ref{thm:gloptparbB}), while the transmission capacity at SNE is $\lambda p \exp\left(- \frac{p\lambda}{\lambda_{max}}\right)$ (Theorem \ref{thm:SNEfinite}).
\end{proof}

{\it Discussion:} In this section, we considered the game theoretic setting for ALOHA MAP where each transmitter unilaterally tries to maximize its own throughput. With each transmitter harvesting energy from nature at a finite rate, we showed that the selfishly optimal and globally optimal strategies are not very different, and the PoA is quite small. With conventional energy powered transmitters, the selfish strategy is to always transmit, however, in the energy harvesting setting there is no incentive for any transmitter to transmit aggressively, since more transmission attempts deplete the energy available for future transmission. 

In the previous two sections we analyzed the transmission capacity of an ad hoc network with ALOHA MAP. Another widely used MAP is CSMA, and in the next section we analyze the performance of CSMA in an ad hoc network with energy harvesting nodes.
\section{CSMA}
In this section, we consider the CSMA MAP and consider a slightly different network model compared to Section \ref{sec:sys}, that has been introduced in \cite{Kaynia2008,Kaynia2011}. In Section \ref{sec:sys}, we assumed that transmitter locations are distributed as a 2-D PPP, and each transmitter always had a packet to transmit and uses ALOHA MAP for packet transmissions. With CSMA, analyzing such model is rather challenging. In this section, we consider an area $A$, and model the packet arrival process as a one-dimensional PPP with arrival rate $(A/L)\lambda$, where $L$ is the fixed packet duration.
Each packet after arrival is assigned to a transmitter location that is uniformly
distributed in area $A$, and the receiver corresponding to a particular transmitter is located at a fixed
distance $d$ away with a random orientation, as shown in Fig. \ref{fig:csma}. For $A \rightarrow \infty$, with ALOHA MAP, this process corresponds to a 2-D PPP of transmitter locations with density $\lambda$ (Section \ref{sec:sys}), where each transmitter has packet 
arrival rate of $\frac{1}{L}$.

The SIR between transmitter $T_n$ and its receiver 
$R_n$ at time $t$ is  $\SIR_n(t) = \frac{ 
d^{-\alpha} |h_{nn}|^2
}
{
\sum_{
{\cal T}_s \in \Phi \backslash \{ {\cal T}_n\}
}
{\mathbf 1}_{T_s}(t) 
d_{ns}^{-\alpha}|h_{ns}|^2}$, 
similar to Section \ref{sec:sys}, where $\Phi$ is the set of all transmitters, and ${\mathbf 1}_{T_s}(t) =1$, if the transmitter $T_s$ is not in back-off and has energy to transmit, and $0$ otherwise. 
With CSMA MAP, transmitter $T_n$ sends its packet at time $t$ if the channel is sensed {\it idle} at time $t$, which in our case corresponds to $\SIR_n(t) > \theta$, with unit power  if available energy $E_n(t) \ge 1$. Otherwise, the transmitter backs off and makes a retransmission attempt after a random amount of time. If $T_n$ transmits the packet, the packet transmission can still fail if $\SIR_n$ falls below $\theta$ for the duration of packet transmission $L$. Thus, the outage probability $P_{out} =  P_b + (1-P_b) P_{\text{fail} | \text{no backoff}}$, where $P_b$ is the back off probability, and $P_{\text{fail} | \text{no backoff}}$ is the probability that the transmission fails. Hence, the transmission capacity with CSMA MAP is defined as 
$C = \lambda  (1-P_{out}) R$ bits/sec/Hz/m$^2$.

Similar to Section \ref{sec:sys}, we assume that the energy 
arrival process is i.i.d. Bernoulli with rate $p$ across different transmitters. In this section, we only consider the  $B=\infty$ case. Analysis for finite $B$ follows similarly. The transition probability diagram for energy queue with CSMA is identical to  Fig. \ref{fig:BD} with $q$ replaced by $1-P_b$, and $r = P(E_n(t) \ge 1)  = \min\left\{\frac{p}{1-P_b},1\right\}$.

\begin{rem}
CSMA MAP introduces correlation among different transmitter's back-off events, and hence the number of simultaneously active transmitters on the 2-D plane no longer follows a PPP. Nevertheless, for analytical tractability, as an approximation we assume that the transmitter back-off events are independent, and simultaneously active transmitter locations are still PPP distributed. The simulation results show that this assumption is reasonable \cite{Kaynia2008,Kaynia2011}.
\end{rem}
In the next Theorem we derive the back-off probability for any transmitter with the CSMA MAP.

\begin{thm}\label{thm:backoff} The backoff probability $P_b = 1- \exp\left(-\frac{\lambda  }{\lambda_{max}} \right)$ if $ \frac{- \lambda_{max} \ln p}{\lambda} >p$, otherwise $P_b$ satisfies  $P_b = 1-\exp\left(-\frac{\lambda (1-P_b)}{\lambda_{max}} \right)$ which can be solved using Lambert's function $W_0(.)$.
\end{thm}
\begin{proof} Transmitter $T_n$ goes into backoff at time $0$ if 
$\SIR_n$ at time $0$ is less than $\theta$. Note that the set of transmitters active at time $0$ are those that started transmitting between $-L$ to $0$ since the packet length is $L$. The transmitters that become active at any time $t$ between time $-L$ and $0$ is a PPP with density $\frac{\lambda}{L} (1-P_b) r$.
Assuming independent back off events across different transmitters, the active set of transmitters at any time slot between $-L$ and $0$ are independent, and since the union of independent PPPs is also a PPP with sum of the densities, transmitters that are active at time $0$ is a PPP with density $\sum_{i=-L}^{0}\frac{\lambda}{L} (1-P_b) r = \lambda (1-P_b) r$. For large $A$, this translates to having PPP distributed active transmitter locations on the 2-D plane with density $\lambda (1-P_b) r$. 
Thus, $P_b = P(\SIR_n(0) < \theta) = 1-\exp\left(-\frac{\lambda (1-P_b) r}{\lambda_{max}} \right)$ \cite{Baccelli2006}.

Next, we proceed using contradiction. 
Let $P_b > 1-p$. Then $r=1$, and hence $P_b = 1- \exp\left(-\frac{\lambda (1-P_b) }{\lambda_{max}} \right) > 1-p$, which results in $1-P_b \ge \frac{-\lambda_{max} \ln p}{\lambda}$. However, $p, \lambda$, and $\lambda_{max}$ are fixed parameters and if they satisfy the relation $\frac{- \lambda_{max} \ln p}{\lambda} >p$, it implies that $P_b \le 1-p$. Thus, we get a contradiction, 
since we started with $P_b > 1-p$. Hence if $\frac{- \lambda_{max} \ln p}{\lambda} >p$, $P_b \le 1-p$, and correspondingly $r = \frac{p}{1-P_b}$ and 
$P_b = 1- \exp\left(-\frac{\lambda  }{\lambda_{max}} \right)$. The other case is obvious.\end{proof}
Next, we derive an explicit expression for packet failure probability with the CSMA MAP.
\begin{thm}\label{thm:outage}  $P_{\text{fail} | \text{no backoff}} = 1- \frac{\sum_{\ell=0}^{L+1}(-1)^{\ell} \binom{L+1}{\ell} \exp^{-\frac{\lambda}{T} \left( \int_{\bbR^2} 1-  \left(\frac{(1-P_b)r}{1+d^{\alpha}\theta x^{-\alpha}} + 1- (1-P_b)r\right)^{\ell} dx\right)}}{1-P_b}$. For $L=1$, $P_{\text{fail} | \text{no backoff}} = 1-(1-P_b)\exp\left(2 \lambda \theta^{2/\alpha} d^2 (1-P_b)^2r^2 \pi^2 \frac{\alpha-2}{\alpha} csc\left(\frac{2\pi}{\alpha}\right)  \right)$.
\end{thm}
\begin{proof} Note that $P_{\text{fail} | \text{no backoff}}$ is the probability that at any time $t$, $\SIR_n(t)<\theta$ for $0< t \le L$ given that $\SIR_n(0) > \theta$. The transmitters that become active at any time $t$ between time $0$ and $L$ is a PPP with density $\frac{\lambda}{L} (1-P_b) r$.
Then,
\begin{eqnarray*}
P_{\text{fail} | \text{no backoff} } &=& 1- P(\SIR_n(1)> \theta, \dots, \SIR_n(L) > \theta | \SIR_0 > \theta),\\
&=&1- \frac{P(\SIR_0 > \theta, \SIR_n(1)> \theta, \dots, \SIR_n(L) > \theta)}{P(\SIR_0 > \theta)}, \\
&=& 1- \frac{\sum_{\ell=0}^{L+1}(-1)^{\ell} \binom{L+1}{\ell} \exp^{-\frac{\lambda}{L} \left( \int_{\bbR^2} 1-  \left(\frac{(1-P_b)r}{1+d^{\alpha}\theta x^{-\alpha}} + 1- (1-P_b)r\right)^{\ell} dx\right)}}{1-P_b},
\end{eqnarray*}
where the expression in the numerator follows from \cite{VazeTDR2011}. Moreover, for the special case of $L=1$,  $P(\SIR_1> \theta | \SIR_0 > \theta) = (1-P_b)\exp\left(2 \lambda \theta^{2/\alpha} d^2 (1-P_b)^2r^2 \pi^2 \frac{\alpha-2}{\alpha} csc\left(\frac{2\pi}{\alpha}\right)  \right)$ \cite{Ganti2009}.
%
\end{proof}

Hence using $P_{out} =  P_b + (1-P_b) P_{\text{fail} | \text{no backoff}}$, we get the 
transmission capacity $C = \lambda  (1-P_{out}) R$ for CSMA MAP by combining Theorem \ref{thm:backoff} and \ref{thm:outage}. Finding the closed form expression for $P_{\text{fail} | \text{no backoff}}$ derived in Theorem \ref{thm:outage} is quite challenging. An upper bound on the $P_{\text{fail} | \text{no backoff}}$, however, can be found  using the FKG inequality \cite{Grimmett1980} as follows. 

\begin{defn}\label{exmdecevent} Let $(\Omega, {\cal F}, {\cal P}$) be the probability space. Let $A \in {\cal F}$, and ${\mathbf 1}_{A}$ be the indicator function of $A$. Event $A \in \ {\cal F}$ is called increasing if ${\mathbf 1}_{A}(\omega) \le {\mathbf 1}_{A}(\omega')$, whenever $\omega \le \omega'$ for some partial ordering on $\omega$. The event $A$ is called decreasing if its complement $A^{c}$ is increasing.
\end{defn}
 \begin{lemma}\label{lemfkg}(FKG Inequality ) If both $A, B \in {\cal F}$ are increasing or decreasing events then $P(AB) \ge P(A)P(B)$ \cite{Grimmett1980}.
\end{lemma}

\begin{lemma} The outage probability of CSMA MAP  $P_{out} \le  1- (1-P_b)^{L+1}$.
\end{lemma}
\begin{proof} Since $\SIR_n(t)$ is decreasing function of the number of interferers, by considering $\omega = (a_1, a_2, \ldots, )$ where for $m \in \bbN$, $a_m =1$ if transmitter $T_m$ is active, and $0$ otherwise, it follows that the success event $\{\SIR_n(t) >\theta\}$ is a decreasing event. Hence, from the FKG inequality, $P(\SIR_n(0) > \theta, \SIR_n(1)> \theta, \dots, \SIR_n(L) > \theta) \ge P(\SIR_0 > \theta)^{L+1}$, since $\SIR_n(t)$ is identically distributed for any $t$. Hence, $P_{\text{fail} | \text{no backoff} } \le 1- (1-P_b)^{L}$, and $P_{out} \le 1- (1-P_b)^{L+1}$.
\end{proof}

{\it Discussion:} In this section, we considered the CSMA MAP for an ad hoc network with energy harvesting nodes. We derived expressions for back-off and outage probability for the CSMA MAP, thereby characterizing the transmission capacity. We  showed that 
depending on the rate of energy arrival  $p$, back-off probability can be written in closed form or can be expressed in terms of  Lambert's function. We also derived an exact expression (and a simplified lower bound) for the outage probability, to characterize the transmission capacity with CSMA MAP.

\vspace{-.1in}
\section{Simulations} In this section, we present some numerical examples to illustrate our theoretical results.
In all simulations we use energy arrival rate $p=.5$, $\alpha =3$, $\theta =2$, and $d=2$, such that $\lambda_{max} = .023$, except for Fig. \ref{fig:TCinfBg}, where $\alpha =3$,  $\theta =1$, and $d=1$, and 
$\lambda_{max} = .2632$ is used.
In Figs. \ref{fig:TCinfBl} and \ref{fig:TCinfBg}, we plot the transmission capacity for ALOHA MAP for 
$B=\infty$ with $\frac{\lambda_{max}}{\lambda} > p$ and $\frac{\lambda_{max}}{\lambda} \le p=0.5$, respectively, from which we can see that for $\frac{\lambda_{max}}{\lambda} \le p$, the optimal transmission probability is $q^{\star}=\frac{\lambda_{max}}{\lambda}$, while in the other case $p\le q^{\star}\le 1$, as derived in Theorem \ref{thm:gloptp}.
In Figs. \ref{fig:TCB1} and \ref{fig:TCB5}, we plot the transmission capacity for ALOHA MAP with  $B=1$, and $B=5$, respectively. From Figs. \ref{fig:TCinfBl}, \ref{fig:TCB1}, and \ref{fig:TCB5}, we can see that as $B$ increases $q^{\star}$ goes from $0.3$ for $B=1$ to $0.23$ for $B= \infty$ for fixed set of parameters. 
Finally, in Fig. \ref{fig:backoff}, we plot the back-off probability with CSMA MAP as a function of $\lambda$. Following Theorem \ref{thm:backoff}, we see that for $\lambda =.01$ and $.035$ for which $\frac{-\ln(p)\lambda_{max}}{\lambda} >p$, back off probability is equal to $1-\exp^{-\left(\frac{p\lambda}{\lambda_{max}}\right)}$, while for $\lambda =.05$ and $.1$, where $\frac{-\ln(p)\lambda_{max}}{\lambda} \le p$, it satisfies the equation $P_b = 1- \exp^{-\left(\frac{r(1-P_b)\lambda}{\lambda_{max}}\right)}$.

\section{Conclusions} In this paper we considered ALOHA and CSMA MAP for an ad hoc network, and derived optimal transmission probability for ALOHA MAP, and back-off and outage probability expressions for CSMA MAP, when each node in the network harvest energy from nature. We characterized the dependence of system throughput on the energy arrival rate, and derived system parameters for optimal performance. In this work, we assumed that each transmitter attempts to transmit with same probability irrespective of the current energy state. With finite battery capacity, it makes more sense to transmit aggressively in high energy states and vice-versa. 
Analyzing energy aware transmission strategies remains an important problem to solve in future with energy harvesting nodes.

\appendices

\section{}\label{app:nondec}
Proof of Lemma \ref{lem:inc}: We first consider the case of $B=1$. For $B=1$, $f_1(q) = \frac{pq}{p+q-pq}$. Hence $f_1'(q) = \frac{p^2}{(p+q-pq)^2}$ and hence $f_1'(q) > 0$ for $q \in [0,1]$. 
For $B>1$, for $p\ne q$, we next show that $f_B'(q)>0$ for $q\in [0,p) \cup (p,1]$. 
Let $c \bydef \left(\frac{p}{1-p}\right)^B$. For $B >1$, $q\ne p$ we compute the first derivative of $f_B(q)$ as $f_B'(q) =\frac{c (1 - q)^{B-1} (pc (1 - q)^{B+1} +
   q^B ( B q + p (-1 - B + q)))}{(q^{B+1}- pc (1 - q)^B)^2} >  0 \ \text{for} \ q  \in [0,p) \cup (p,1]$. 
   In the interest of space we do not provide intermediate steps and only write the final answer.
Also note that for $q=p$ with $B>1$, $f_{B}(q) = rq = \frac{qB}{B+1-q}$. It can be checked that $f_{B}(p-\delta) < f_{B}(p) < f_{B}(p+\delta)$, for small $\delta>0$. Hence $f_{B}(q)$ is an increasing function of $q \in [0,1]$.

\bibliographystyle{../../../IEEEtran}
\bibliography{../../../IEEEabrv,../../../Research}
\newpage
\begin{figure}
\centering
\includegraphics[width=4in]{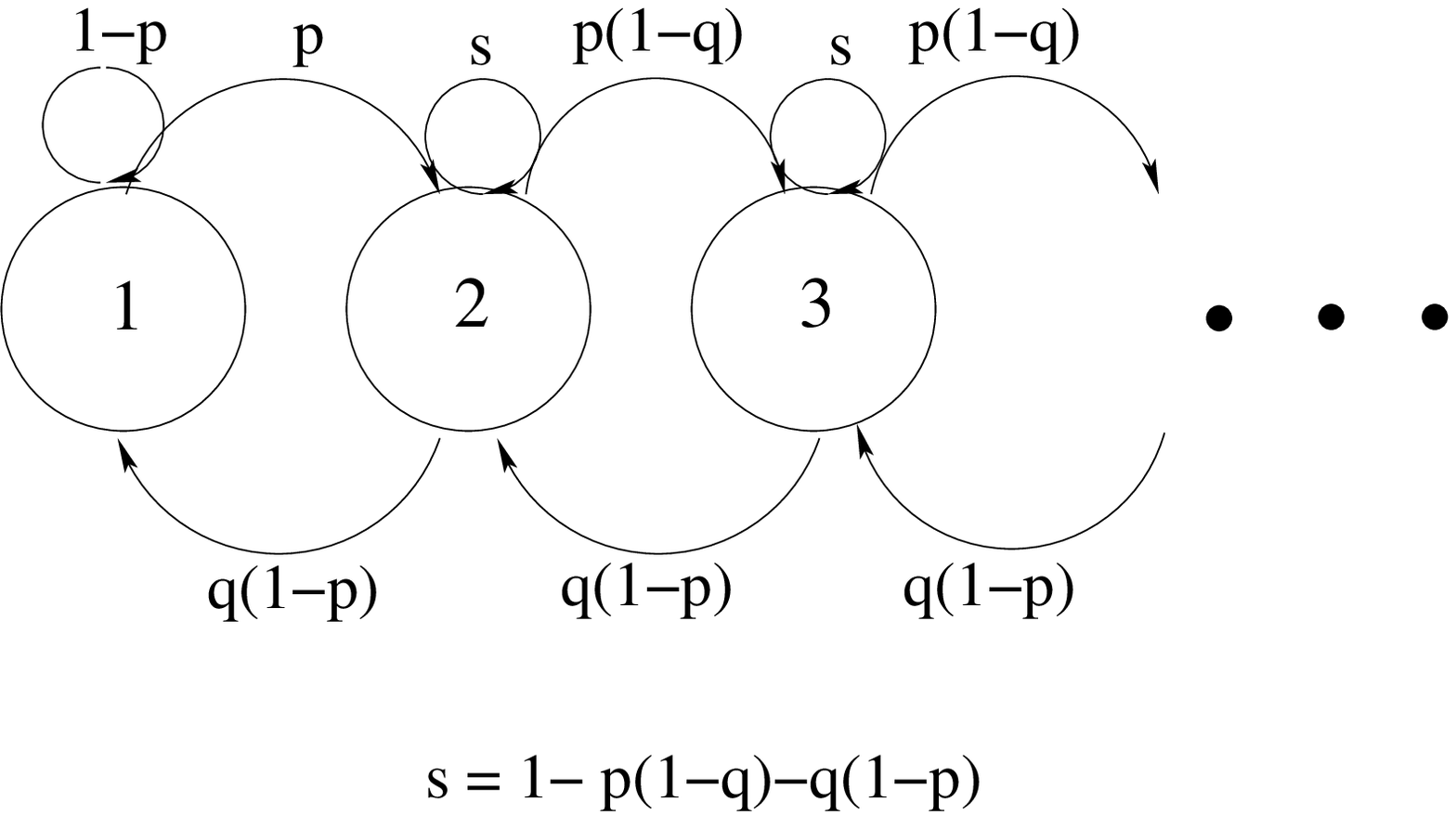}
\caption{Transition state probabilities of infinite state birth-death Markov process. }
\label{fig:BD}
\end{figure}

\begin{figure}
\centering
\includegraphics[width=4in]{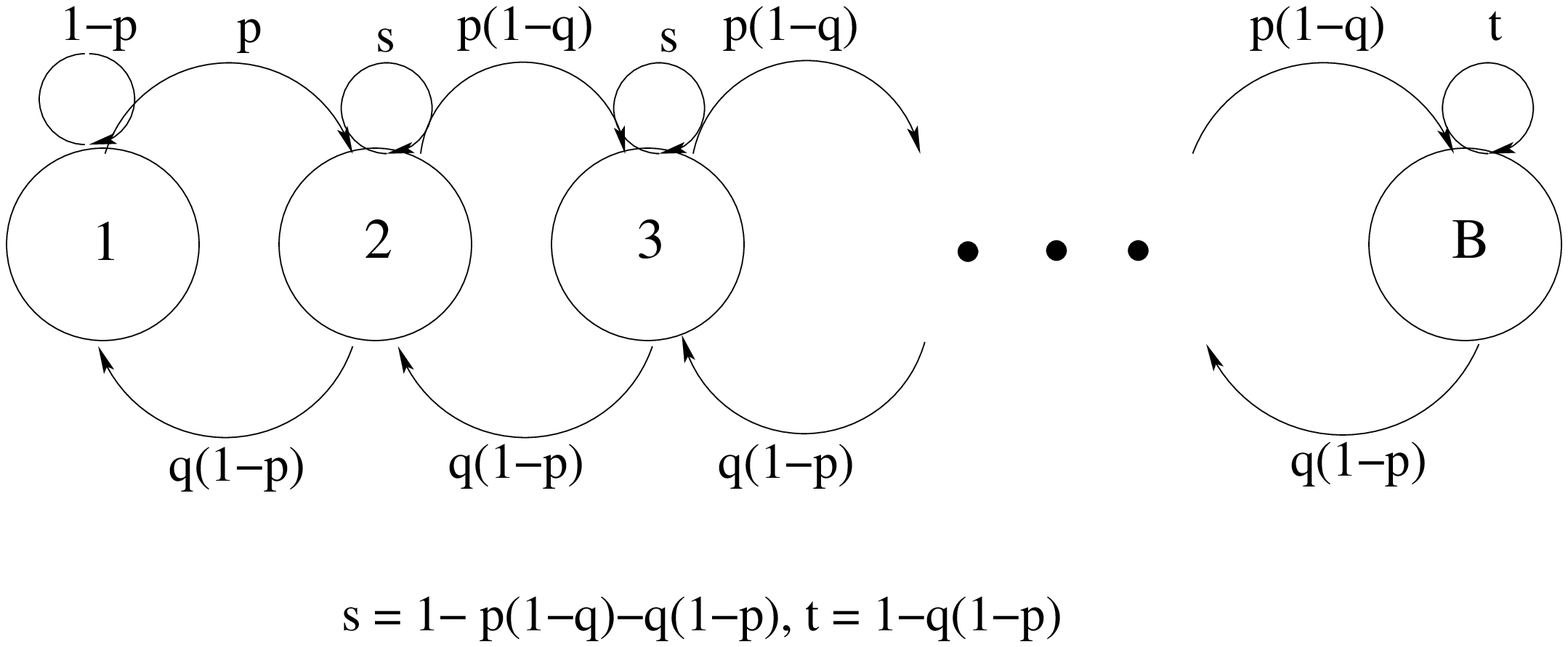}
\caption{Transition state probabilities of finite state birth-death Markov process. }
\label{fig:maxBD}
\end{figure}

\begin{figure}
\centering
\includegraphics[width=4in]{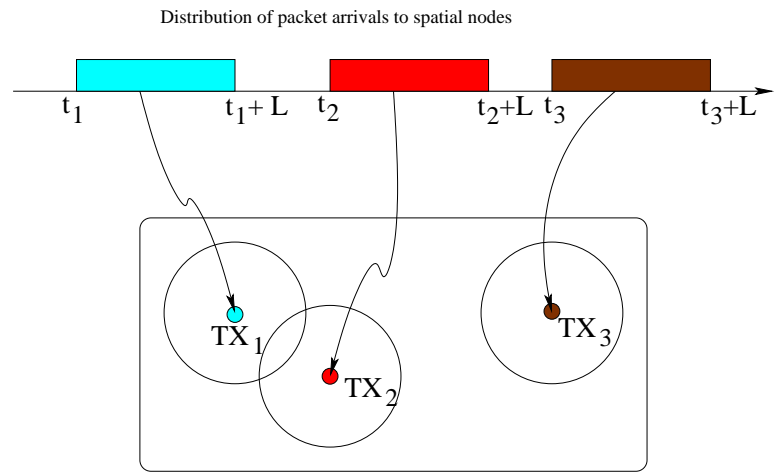}
\caption{Packet arrival model for CSMA MAP. }
\label{fig:csma}
\end{figure}

\begin{figure}
\centering
\includegraphics[width=4in]{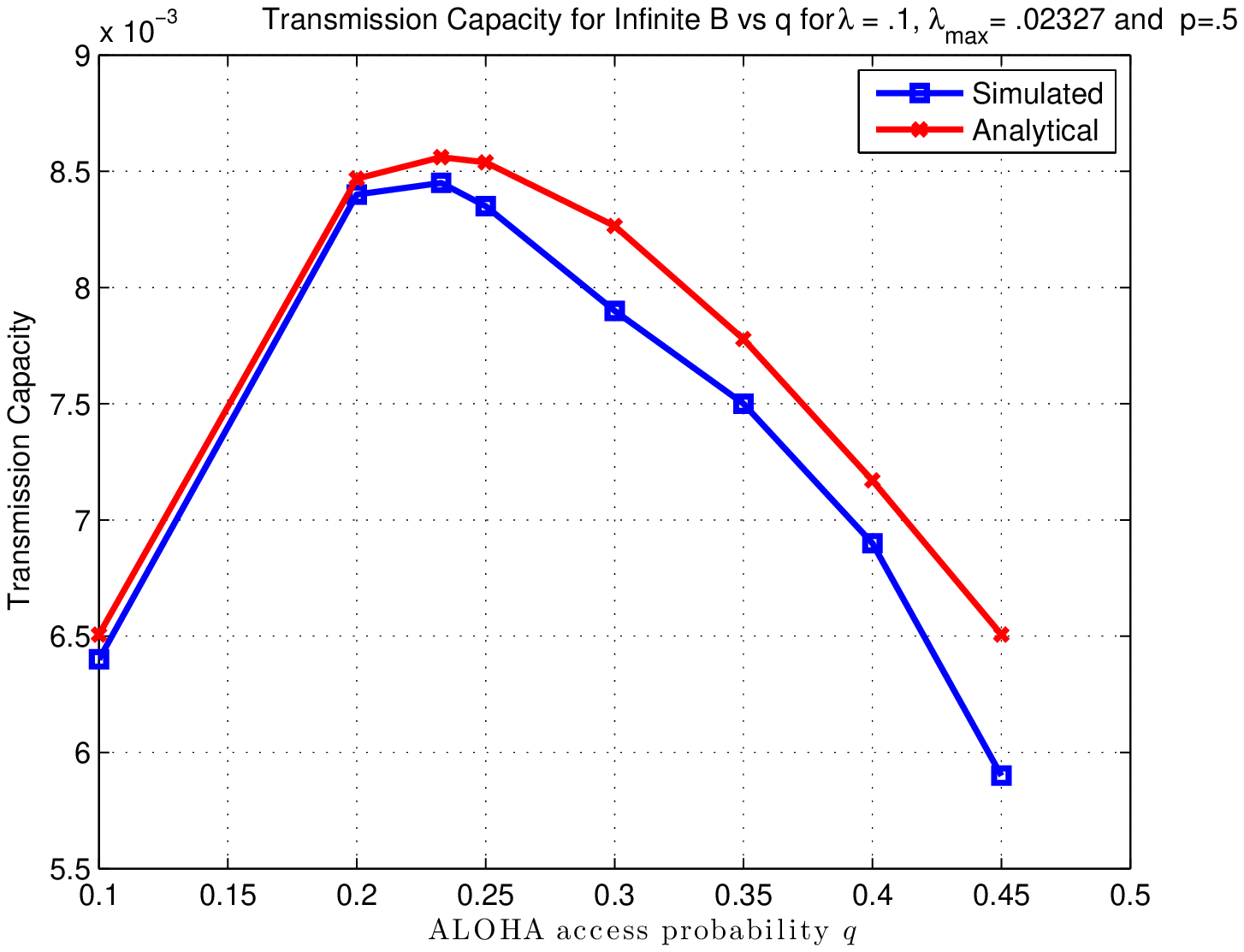}
\caption{Plot of transmission capacity for $B=\infty$ with $\frac{\lambda_{max}}{\lambda} \le p=0.5$.  }
\label{fig:TCinfBl}
\end{figure}

\begin{figure}
\centering
\includegraphics[width=4in]{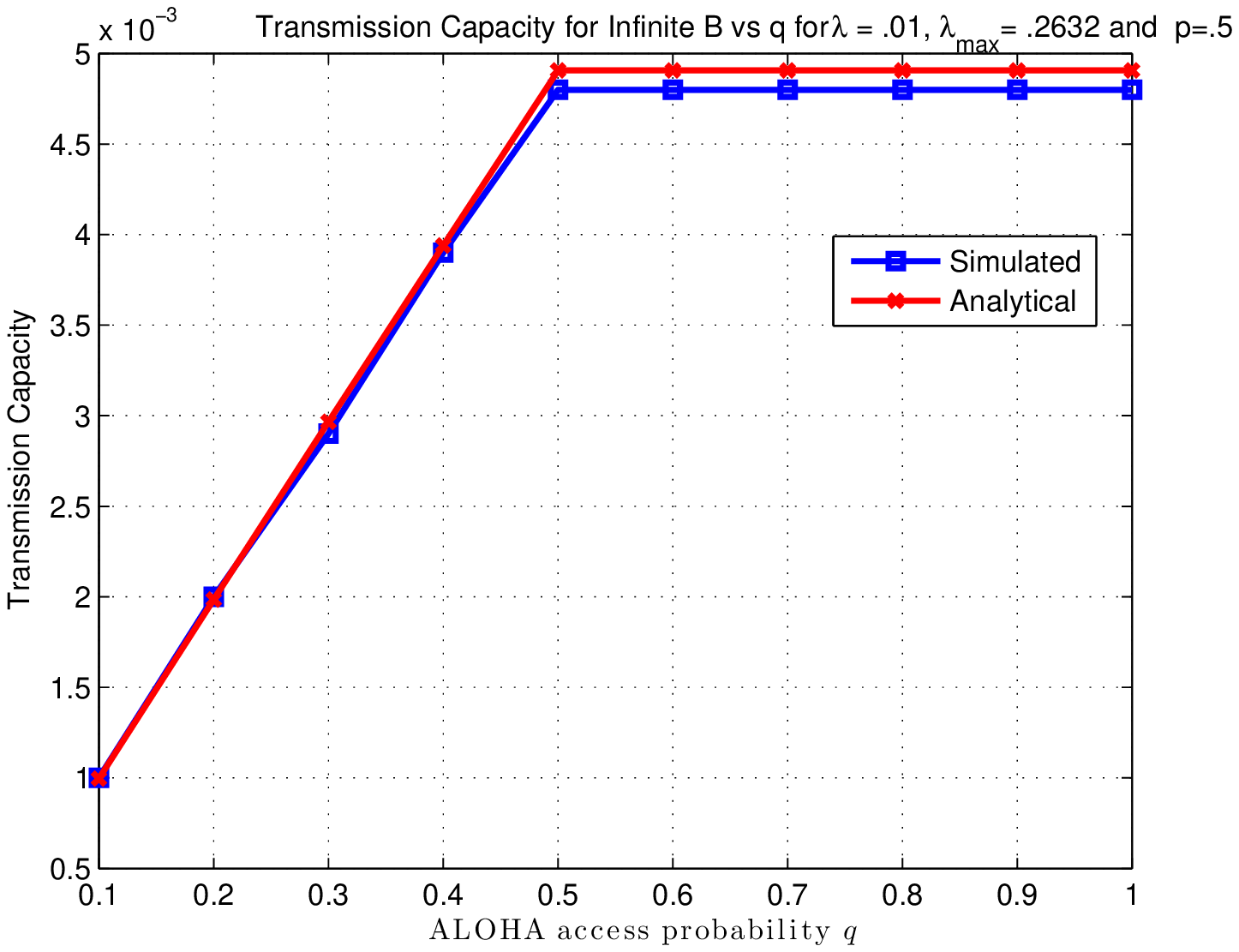}
\caption{Plot of transmission capacity for $B=\infty$ with $\frac{\lambda_{max}}{\lambda}>p=0.5$.  }
\label{fig:TCinfBg}
\end{figure}

\begin{figure}
\centering
\includegraphics[width=4in]{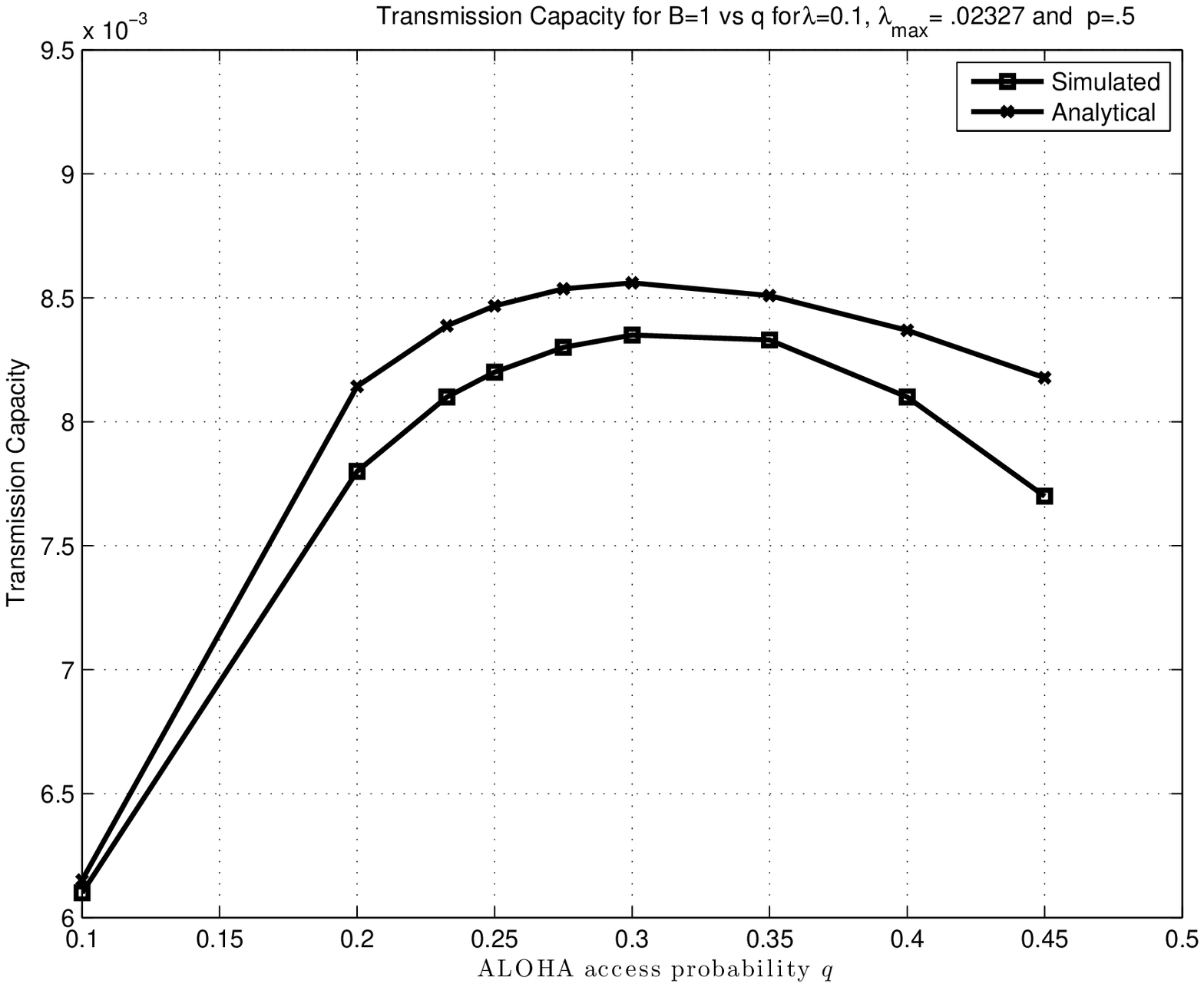}
\caption{Plot of transmission capacity for $B=1$ with  $p=0.5$.}
\label{fig:TCB1}
\end{figure}

\begin{figure}
\centering
\includegraphics[width=4in]{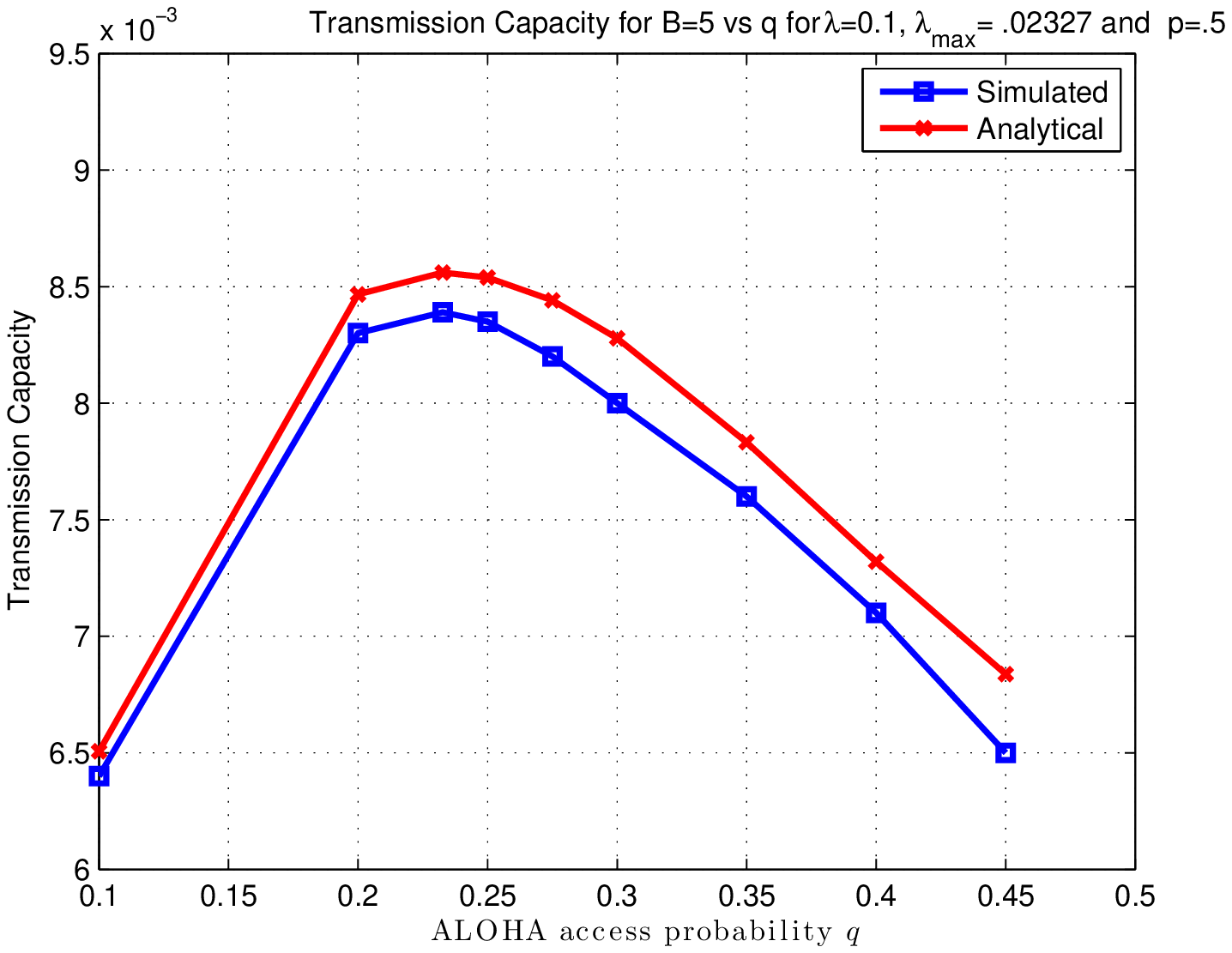}
\caption{Plot of transmission capacity for $B=5$ with  $p=0.5$. }
\label{fig:TCB5}
\end{figure}

\begin{figure}
\centering
\includegraphics[width=4in]{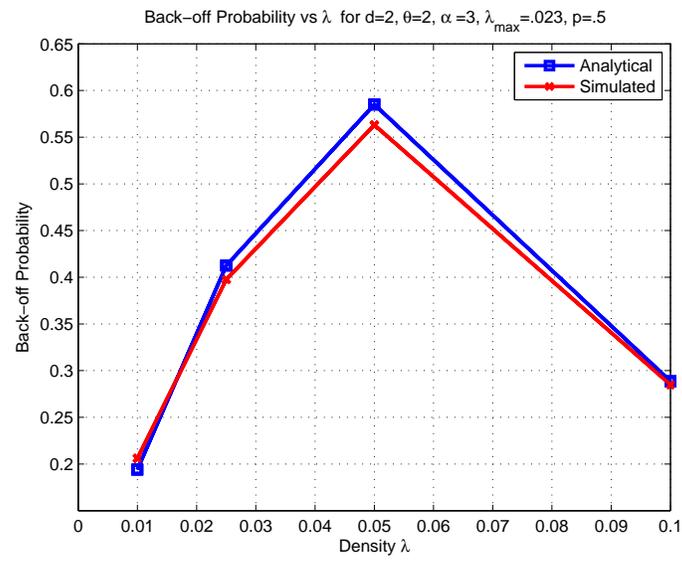}
\caption{Backoff probability as a function of $\lambda$.}
\label{fig:backoff}
\end{figure}

   \end{document}